\documentclass[12pt]{article} 
\usepackage[sectionbib]{natbib}
\usepackage{array,epsfig,fancyhdr,rotating}
\usepackage[]{hyperref}  
\usepackage{sectsty, secdot}
\usepackage[normalem]{ulem}
\usepackage{orcidlink}
\usepackage{lmodern}
\sectionfont{\fontsize{12}{14pt plus.8pt minus .6pt}\selectfont}
\renewcommand{\theequation}{\thesection\arabic{equation}}
\subsectionfont{\fontsize{12}{14pt plus.8pt minus .6pt}\selectfont}
\usepackage{geometry}
\usepackage{setspace}
\usepackage{amsmath}
\usepackage{amssymb}
\usepackage{color}
\usepackage{amsfonts}
\usepackage{multirow}
\usepackage{mathdots}
\usepackage{amsthm}
\usepackage{bm}
\newtheorem{theorem}{Theorem}

\newtheorem{proposition}{Proposition}
\theoremstyle{definition}

\newtheorem{remark}{Remark}
\newtheorem{assumption}{Assumption}
\let\customlarge\large
\renewcommand{\customlarge}{\fontsize{13.5pt}{13pt}\selectfont}

\begin{document}


\renewcommand{\baselinestretch}{2}

\markright{ \hbox{\footnotesize\rm Statistica Sinica
}\hfill\\[-13pt]
\hbox{\footnotesize\rm
}\hfill }

\markboth{\hfill{\footnotesize\rm GUANNAN ZHAI AND FEIFANG HU} \hfill}
{\hfill {\footnotesize\rm VALID INFERENCE FOR MULTI-ARM TRIALS} \hfill}

\renewcommand{\thefootnote}{}
$\ $\par


\fontsize{12}{14pt plus.8pt minus .6pt}\selectfont \vspace{0.8pc}
\centerline{\customlarge\bf Valid test for multi-arm trials with generalized }
\vspace{2pt} 
    \centerline{\customlarge\bf linear models under covariate-adaptive randomization }
\vspace{.4cm} 
\centerline{
Guannan Zhai\,\orcidlink{0009-0004-6773-9640}$^{1}$,
Feifang Hu\,\orcidlink{0000-0002-9811-2910}$^{1}$
}
\vspace{.4cm} 
\centerline{\it $^{1}$Department of Statistics, George Washington University, U.S.A.}
 \vspace{.55cm} \fontsize{9}{11.5pt plus.8pt minus.6pt}\selectfont


\begin{quotation}
\noindent {\it Abstract:}
Modern medical research, such as dose-finding studies, seamless trials, and shared control designs, often involves comparing multiple treatments simultaneously. Despite its wide applications, most research focuses on continuous endpoints, leaving the inference for general outcome types in high demand. In this article, we propose a new inference method for conducting multiple-treatment comparisons involving endpoints within the generalized linear model (GLM) framework under covariate-adaptive randomization (CAR). First, we investigate the asymptotic properties of the standard \textcolor{black}{Wald z-statistics (z-scores)} in multi-arm trials, highlighting issues when the working model is misspecified, particularly through omitted covariates. Our theoretical findings reveal that these \textcolor{black}{z-scores} do not consistently converge to a standard multivariate normal distribution, leading to either conservative or inflated Type I error rates, depending on the specific GLM endpoint. Second, based on these theoretical results, we develop adjusted test statistics to correct the distributional problems. To appropriately control the family-wise Type I error rate inherent in multi-arm comparisons, we incorporate our adjusted statistics with Simes-type multiple-testing procedures. This robust inference method can effectively control Type I error while potentially improving power. Extensive simulation studies and a real-world application to a metastatic breast cancer trial confirm the effectiveness and practicality of our approach.

\vspace{9pt}
\par
\end{quotation}\par

\def\thefigure{\arabic{figure}}
\def\thetable{\arabic{table}}

\renewcommand{\theequation}{\thesection.\arabic{equation}}

\fontsize{12}{14pt plus.8pt minus .6pt}\selectfont

\section{Introduction}

Clinical trials are widely utilized in medical research to evaluate treatment efficacy and establish causal relationships between interventions and health outcomes. Traditional clinical trials typically employed simple two-arm comparisons with clearly defined continuous endpoints. However, these conventional designs have become inadequate for the complexity of modern medical interventions and increasing demand for trial efficiency and adaptability. Consequently, modern clinical trials have evolved towards more sophisticated designs. For example, dose-finding studies, seamless trials, and shared control designs are widely employed, often simultaneously comparing multiple treatments. Despite their wide application, statistical inference methods remain predominantly focused on two-arm trials with continuous outcomes, leaving methods for general outcomes in multi-arm settings underdeveloped.

Several studies have begun addressing inference issues in complex trial designs under covariate-adaptive randomization (CAR). CAR adaptively assigns patients to treatment groups based on previous patient allocations and their baseline covariates (e.g., age, gender, disease severity, pre-treatment). Its randomization framework and associated theoretical properties have been extensively studied, providing robust foundations for their use in clinical trials \citep{pocock1975sequential, HuHuASYMPTOTIC, hu2020theory, liu2022balancing, hu2023multi, ma2023carat, ma2024new, zhao2024consistent}. Consequently, CAR methods are now widely adopted in practice to ensure efficient and precise treatment comparisons.

Building upon CAR’s extensive theoretical foundation, recent research has started addressing inference issues in multi-arm trials. \textcolor{black}{For instance, \citet{bugni2019inference} extend their earlier work on two-arm trials \citep{bugni2018inference} to multi-arm settings. They develop robust OLS-based inference for average treatment effects under stratified randomization, using linear working models such as fully saturated and strata fixed-effects regressions.} Additionally, \citet{ma2022seamless} investigated adjusted inference methods for multi-arm seamless phase II/III designs with continuous endpoints under CAR. Related work for non-continuous outcomes under CAR has mostly focused on two-arm trials \citep{ma2022regression, shao2013validity, yi2020cox, ye2020robust, li2021}. \textcolor{black}{Overall, robust GLM-based inference for multi-arm trials with non-continuous endpoints, especially when randomization covariates are omitted from the analysis model, remains underdeveloped.}

Currently, a critical gap remains in statistical methodologies for multi-arm trials with non-continuous outcomes. A particularly important class of such non-continuous endpoints includes binary responses and count data, which are fundamental to clinical research and are usually analyzed within the GLM framework. The methodological gap stems from two fundamental challenges. First, unlike two-arm trials, multi-arm trials typically employ a shared control group, creating dependencies among treatment comparisons that invalidate traditional statistical approaches. Second, when analyzing data under CAR where the endpoints accommodate GLMs, the common practice of omitting covariates from the working model raises issues about the validity of standard inference methods.

To address these challenges, we develop a comprehensive framework for valid statistical inference in multi-arm trials with GLM-accommodating endpoints, particularly focusing on scenarios where the working model omits covariates. Our work makes three main contributions: (1) theoretically investigating the asymptotic properties of the standard Wald test under CAR in multi-arm trials with shared controls; (2) developing valid statistical inference methods that properly account for both the shared control structure and scenarios where working models cannot include all randomization covariates; and (3) establishing theoretical guarantees for the proposed methods under GLM settings. Our framework provides the first theoretical foundation for statistical inference in this important yet under-explored setting, while maintaining practical applicability in real clinical trials. 

The paper is organized as follows. Section \ref{sec: notation} introduces the GLM-based testing framework, including notation, working model, test statistic, and Type I error control. Section \ref{sec: theorem} presents our main theoretical results, examining test size under commonly used models (logistic, linear, Poisson, and exponential regression) and exploring extensions to non-canonical link functions. Section \ref{sec: new test} introduces an adjusted test that maintains valid Type I error control. Sections \ref{sec: simu} and \ref{sec: realtrial} evaluate our method through simulation studies and an application to a metastatic breast cancer trial, respectively. Section \ref{sec:discussion} concludes with a discussion of our findings.

\section{ Test under Generalized Linear Models}
\label{sec: notation}
\subsection{Framework and Notation} 
\label{subsec:notations and framework}
We consider a clinical trial with $K$ treatment groups ($K \geq 1$) and one control group, where a total of $N$ subjects are allocated to one of $(K+1)$ arms based on their covariate profiles via a multi-arm CAR procedure proposed by \cite{hu2023multi}. For each subject $i$, let $T_{ik}$ be the treatment assignment indicator, where $T_{ik} = 1$ if subject $i$ is assigned to arm $k$, and $0$ otherwise. Here, $k = 0$ denotes the control group and $k \in \{1, \ldots, K\}$ denotes the treatment groups. Let $\bm{Z}_i = (Z_{i1}, \ldots, Z_{iq})$ represent the covariate profile of subject $i$, and $y_i$ be their observed response under the assigned treatment. The true model is assumed to be: 
\begin{align}
\label{mod: true}
    \mathbb{E}\left[Y_i \mid \bm{T}_i, \bm{Z}_i \right] = h\left(\delta_0 + \delta_1 T_{i1} + \delta_2 T_{i2} + \ldots + \delta_K T_{iK} + \bm{\beta}^T\bm{Z}_i\right)
\end{align}
where $\bm{T}_i$ is the assignment vector of subject $i$, $\delta_k$ is treatment effect of group $k$, $\bm{\beta} = (\beta_1, \ldots, \beta_q)$ are coefficients for prognostic factors $\bm{Z}_i$. 

In the GLM framework considered in this paper, the function $h(\cdot)$ denotes the inverse of a monotonic link function, establishing the relationship between the linear predictor and the expected outcome. For example, $h(u)=u$ corresponds to linear regression models suitable for continuous outcomes, while $h(u)=\exp (u) /\{1+\exp (u)\}$ corresponds to logistic regression models commonly used for binary responses. Similarly, $h(u)=\exp (u)$ corresponds to Poisson regression models suitable for count data. Our theoretical results and proposed inference procedures hold under general assumptions on the link function $h(\cdot)$, including canonical as well as certain non-canonical link functions. Specific assumptions and conditions on $h(\cdot)$ required for our theoretical results are detailed in later sections.

The framework we discuss in this article is based on the following assumptions. 
\begin{assumption}
\label{asu: covariate independence}
    Throughout the article, we assume the covariates are independent and identically distributed for each patient. 
\end{assumption}
\begin{assumption}
\label{asu: covariate center mean}
    All covariates are centered with $\mathbb{E}(\bm{Z}) = \bm{0}$.
\end{assumption}
\begin{assumption}
\label{asu: discrete strata}
The covariates $\bm{Z}_i$ in model \eqref{mod: true} are discrete with
finitely many levels and coincide with the stratification covariates used by
the CAR procedure. Consequently, the number of strata is finite and
$\mathbb{E}(Y_i \mid \bm{Z}_i)$ is constant within each stratum.
\end{assumption}
\begin{assumption}
\label{asu: moments}
$\mathbb{E}(Y^2) < \infty$; the outcomes are conditionally independent across
subjects given the covariates and treatment assignments; and $h(\cdot)$ is
strictly monotone and twice continuously differentiable in a neighborhood of
$h^{-1}\{\mathbb{E}(Y)\}$ with $h'[h^{-1}\{\mathbb{E}(Y)\}] \neq 0$.
\end{assumption}

For the Assumption \ref{asu: covariate center mean}, even though the true mean of covariate $j$ is $c$ ($c \neq 0$), we can center it by subtracting $c$ from each value. \par 
In GLMs, the distribution of $Y$ given the linear predictor $\delta_0 + \delta_1 T_{i1} + \delta_2 T_{i2} + \ldots + \delta_K T_{iK} + \bm{\beta}^T\bm{Z}$ is assumed to follow an exponential family distribution, which has the form  
\begin{align}
    \exp \left\{\frac{Y \theta-b(\theta)}{\phi}+c(Y, \phi)\right\}
\end{align} 
where $\phi$ is the scale parameter, $\theta$ is the canonical parameter of the distribution, and $b(\cdot)$ and $c(\cdot)$ are known functions.  
\subsection{Working Model and Test} 
Before a trial begins, detailed planning is required, including how the qualified patients will be randomized and the statistical methodology for post-randomization data analysis. The post-randomization data analysis usually involves the inference method and construction of the associated test statistic. A well-established statistical method can effectively control the Type I error rate. More specifically, consider a test statistic $S$ constructed from the observed outcome data of patients, where the null hypothesis is rejected if and only if $|S|>z_{\alpha / 2}$. Under $H_0$, a valid well-controlled test statistic maintains a Type I error rate close to nominal significance level $\alpha$. When all strata factors are used in the post-randomization analysis, the  test statistic is well-controlled. \textcolor{black}{However, the covariates used in CAR may be partially or fully omitted from the post-randomization data analysis due to simplicity of a testing procedure and some practical reasons. \textcolor{black}{For example, when $Z_i$ has many levels (e.g., study site/center) or when stratification leads to sparse stratum-by-treatment cells.} In such settings, fully adjusted GLM fitting and standard error estimation can be numerically unstable.} Hence, it is necessary to study the inference properties of the standard tests when a working model with omitted covariates is used. Consider the following working model:
\begin{align}
\label{mod:working}
    \mathbb{E}\left[Y_i \mid \bm{T}_i\right] = h\left(\delta_0 + \delta_1 T_{i1} + \delta_2 T_{i2} + \ldots + \delta_K T_{iK} \right)
\end{align}
where the link function and conditional distribution of $Y$ given linear predictor are the same as in the true model. 

To compare each treatment with the control based on the working model (\ref{mod:working}), we test the individual null hypotheses $H_{0,k}: \delta_k = 0$ for $k = 1,\ldots,K$ and obtain their corresponding p-values. Since multiple comparisons are conducted simultaneously, there is a risk of family-wise Type I error inflation that could lead to false-positive findings. To control this risk and ensure the validity of our results, we implement appropriate multiplicity adjustment methods, which are detailed in Section \ref{subsec:control FWER}. 
The multivariate test can be formed as follows: 
\begin{align}
\label{hyps}
    H_{0k}: \delta_k = 0 \quad \quad k = 1, \ldots, K
\end{align}
with the test statistics 
\begin{align}
\label{Test}
    S_k = \frac{\hat \delta_k}{\text{se}(\hat \delta_k)}
\end{align}

\begin{remark}
\label{rmd:Test}
    The commonly used test statistic for the hypothesis (\ref{hyps}) is normalized to have a unit variance.
\end{remark}

Suppose that $N$ patients have been enrolled in the trial via CAR. For the $i$th patient, we observe $(y_i, T_{ik}, \bm{z}_i)$ as a realization of $(Y, T, Z)$. The log-likelihood function of the observed responses can be written as:

\begin{align}
    \log \prod_{i=1}^N p(y_i \mid \theta_i) = \sum_{i=1}^N \left\{ \frac{\theta_i y_i  - b(\theta_i)}{\phi} + c\left( y_i, \phi\right) \right\}
\end{align}

The maximum likelihood estimator (MLE) takes the form:
\begin{align*}
\hat{\delta}_0 + \hat{\delta}_k = h^{-1}\left( \frac{\sum_{i=1}^N y_iT_{ik}}{n_k}\right)
\end{align*}
From this, we can derive:
\begin{align}
\label{eqn:mle}
    \hat{\delta}_0 &= h^{-1}\left( \frac{\sum_{i=1}^N y_iT_{i0}}{n_0}\right)\\
    \hat{\delta}_k &= h^{-1}\left( \frac{\sum_{i=1}^N y_iT_{ik}}{n_k}\right) -  h^{-1}\left( \frac{\sum_{i=1}^N y_iT_{i0}}{n_0}\right)
\end{align}
where $h^{-1}(\cdot)$ is the link function and $n_k = \sum_{i=1}^N T_{ik}$ denotes the number of patients in group $k$. \\
These estimators form the basis for our hypothesis testing framework, which we next introduce with proper multiplicity adjustment.

\begin{remark}
\textcolor{black}{$\delta_k$ summarizes the treatment effect on the scale of the link function.
For example, with a logit link and binary outcomes, $\exp(\delta_k)$ is the odds ratio comparing the marginal event probabilities under treatment $k$ and control, that is,
$\exp(\delta_k)=\frac{p_k/(1-p_k)}{p_0/(1-p_0)}$ with $p_k=\mathbb{P}\{Y_i(k)=1\}$ and
$p_0=\mathbb{P}\{Y_i(0)=1\}$. With a log link for count outcomes, $\exp(\delta_k)$ is the ratio of marginal mean potential outcomes, namely $\exp(\delta_k)=\mathbb{E}\{Y_i(k)\}/\mathbb{E}\{Y_i(0)\}$. These link-scale contrasts are often used to report efficacy for binary and count endpoints, since they respect natural constraints of the outcome and provide interpretable relative effect measures. When an outcome-scale contrast is of interest, it can be obtained from the implied marginal means, for example $\mathbb{E}\{Y_i(k)\}-\mathbb{E}\{Y_i(0)\}=h(\delta_0+\delta_k)-h(\delta_0)$. In this paper, our focus is valid inference for the prespecified link-scale estimand $\delta_k$ under CAR when covariates are omitted from the working mean model.}
\end{remark}

\subsection{Control of Family-wise Type I error Rate under Multiplicity} \label{subsec:control FWER}
Comparisons across different treatment arms and a control lead to potential inflation of the family-wise Type I error rate as we mentioned previously. For example, in a dose-finding trial, multiple effective doses of a new drug are commonly compared with a single control, resulting in simultaneous testing of multiple null hypotheses, denoted as $H_{01}, H_{02}, \ldots, H_{0K}$. A critical issue here is controlling the family-wise Type I error rate (FWER). Hence, multiplicity adjustment methods that ensure the probability of rejecting at least one true null hypothesis does not exceed the predetermined significance level $\alpha$ are essential \citep{ryan1959multiple, hochberg1987multiple, bauer1991multiple}. Existing literature offers several methods, including Bonferroni, Holm, Simes, and Dunnett adjustments. More advanced approaches are reviewed by \citet{alosh2014advanced}.

In this article, we employ Simes-type procedures
\citep{simes1986improved, hochberg1988sharper, sarkar1997simes}, due to their
simplicity, greater statistical power compared to the Bonferroni and Holm
methods, and their validity under positively dependent test statistics. While
the Dunnett method is specifically designed for multiple comparisons against
a single control group, it requires stronger assumptions such as normality
and equal variance, limiting its applicability in more generalized or
exploratory contexts.

It is important to distinguish two inferential goals. The first is testing
the global null hypothesis $H_0 = \bigcap_{k=1}^K H_{0k}$, that is, detecting
whether any treatment differs from the control. For $K$ null hypotheses with
p-values $p_1, \ldots, p_K$ and ordered values
$p_{(1)} \leq \cdots \leq p_{(K)}$, the Simes test rejects the global null
hypothesis if there exists any $i \in \{1,\ldots,K\}$ such that
\begin{equation}
p_{(i)} \leq \frac{i}{K}\, \alpha .
\end{equation}
The second goal is drawing conclusions about the individual hypotheses
$H_{0k}$ with strong control of the FWER, that is, control under every
configuration of true and false null hypotheses. For this purpose we use
Hochberg's step-up procedure \citep{hochberg1988sharper}, which applies the
Simes critical constants in a step-up manner: let $j^*$ be the largest index
$j$ such that $p_{(j)} \leq \alpha/(K-j+1)$; then reject
$H_{0(1)}, \ldots, H_{0(j^*)}$, and reject nothing if no such $j$ exists.

\begin{proposition}[Validity of Simes-type procedures]
\label{prop:simes-validity}
Suppose that either
\begin{enumerate}
\item[(i)] the test statistics (or the corresponding p-values) are
independent; or
\item[(ii)] the test statistics exhibit the positive regression dependence on
subsets (PRDS) property: for every subset $I \subseteq \{1, \dots, K\}$ of
true null hypotheses and every increasing function $\varphi$, the conditional
expectation
$\mathbb{E}\left[\varphi\left(p_i, i \in I\right) \mid p_j \leq t\right]$
is non-decreasing in $t$ for each $j \notin I$.
\end{enumerate}
Then the Simes test of the global null hypothesis has size at most $\alpha$,
i.e., $P(\text{reject } H_0 \mid H_0 \text{ true}) \leq \alpha$, and
Hochberg's step-up procedure controls the family-wise error rate strongly,
i.e., $P(\text{reject at least one true } H_{0k}) \leq \alpha$ under every
configuration of true null hypotheses, provided the p-values of the true null
hypotheses are valid and satisfy the dependence condition.
\end{proposition}

\noindent\textbf{Proof sketch:}
Under independence, the size of the Simes test was established by
\citet{simes1986improved}. Under positive dependence, the Simes inequality
was extended by \citet{sarkar1997simes} and \citet{sarkar1998some}. The
strong FWER control of Hochberg's procedure follows from applying the Simes
inequality to the subset of true null hypotheses within the closed testing
framework \citep{hochberg1988sharper}.

\bigskip
With these conditions clearly stated, we proceed to present our primary theoretical results in the subsequent sections.

\section{Main Results} \label{sec: theorem}
\textcolor{black}{We place the formal definitions of the imbalance measures used under multi-arm CAR in the
Supplementary Materials. These include the overall imbalance $\bm D_N$, the marginal
imbalances $\bm D_N(r,m_r^*)$, and the within-stratum imbalances $\bm D_N(s)$.
The theorems below only use their stochastic orders and how they enter the asymptotic variance. We keep in the main text only the basic trial notation described in Section \ref{subsec:notations and framework}.}
\subsection{Main Theorem} \label{subsec: Theorems}
Throughout this section, Assumptions \ref{asu: covariate independence}--\ref{asu: moments} are in force.
Theorems \ref{thm2:sigmah} and \ref{thm4:sigmah_noncanonical} additionally
require that the CAR procedure is invariant under permutations of the $K+1$
treatment labels and that $D_a(s) = o_p(N)$ for every arm $a$ and stratum
$s$, as stated in Lemma 2 of the Supplementary Materials.
We now present the asymptotic distribution of the standard Wald test statistic under the working model \eqref{mod:working}. 
\begin{theorem}
\label{thm1}
    Suppose that a covariate-adaptive randomization procedure satisfies the condition that all within-stratum imbalances are bounded in probability; then under the global null hypothesis $H_0: \delta_1 = \cdots = \delta_K = 0$, as $N \to \infty$, we have 
    \begin{align*}
    (S_1, \ldots, S_K)  \xrightarrow{D} \mathcal{N}\left(\bm{0}, \frac{\mathbb{E}[\text{Var}(Y \mid Z)]}{\phi h'[h^{-1}\{\mathbb{E}(Y)\}]}\bm{V}\right)
    \end{align*}
where $\bm{V}$ is a K-dimensional square matrix with $v_{ij} = 1$ if $i = j$ and $v_{ij} = 1/2$ if $i \neq j$. 
\end{theorem}

Based on Theorem \ref{thm1}, if the within-stratum imbalances are bounded in probability, the size of the test for treatment effect depends on a comparison between $\mathbb{E}[\text{Var}(Y\mid Z)]$ and $\phi h^{\prime}\left[h^{-1}\{E(Y)\}\right]$, i.e., in asymptotic sense, 
\begin{itemize}
    \item[-] The test is conservative if $\mathbb{E}[\text{Var}(Y \mid Z)] < \phi h'[h^{-1}\{\mathbb{E}(Y)\}]$. 
    \item[-] The test is valid if $\mathbb{E}[\text{Var}(Y \mid Z)] = \phi h'[h^{-1}\{\mathbb{E}(Y)\}]$
    \item[-] The test is inflated if $\mathbb{E}[\text{Var}(Y \mid Z)] > \phi h'[h^{-1}\{\mathbb{E}(Y)\}]$
\end{itemize}
Besides, the test statistics $S_1, S_2, \ldots, S_K$ are dependent because the same control group is used. 

We assume two key conditions which are also stated in \cite{li2021}:\\
Condition 1: The overall imbalance is bounded in probability $\bm{D}_N = O_p(\bm{1})$ \\ 
\noindent
Condition 2: \noindent
Under the global null hypothesis $H_0$, for each
$k=1,\ldots,K$, as $N\to\infty$,
\[
\frac{1}{\sqrt N}
\sum_{i=1}^N
(T_{ik}-T_{i0})\nu_i
\xrightarrow{D}
\mathcal N(0,\sigma_h^2),
\]
where
\[
\nu_i
=
\mathbb E(Y_i\mid\bm Z_i)
=
h(\delta_0+\bm\beta^T\bm Z_i).
\]
The $K$ treatment-control contrasts are assumed to converge jointly to a multivariate normal distribution.\\
These conditions are generally satisfied by CAR procedures that balance discrete covariates. For CAR designs where within-stratum imbalances are bounded in probability, the distribution in Condition 2 reduces to $O_p(1/\sqrt{N})$, implying $\sigma_h^2 = 0$. It has been proved that these conditions hold for Pocock and Simon's minimization, where the within-stratum imbalances, scaled by $1/\sqrt{N}$, follow an asymptotic joint normal distribution with positive variances \citep{hu2020theory, hu2023multi}. We now present Theorem \ref{thm2:sigmah}. 

\begin{theorem}
\label{thm2:sigmah}
Suppose that a covariate-adaptive randomization procedure satisfies Conditions 1 and 2. Then, under the global null hypothesis
$H_0:\delta_1=\cdots=\delta_K=0$, as $N\to\infty$,
\[
(S_1,\ldots,S_K)
\xrightarrow{D}
\mathcal{N}\left(
\bm{0},
\frac{
\mathbb{E}\!\left[\operatorname{Var}(Y\mid Z)\right]
+(K+1)\sigma_h^2/2
}{
\phi
h'\!\left[
h^{-1}\{\mathbb{E}(Y)\}
\right]
}
\bm{V}
\right),
\]
where $\bm{V}=(v_{ij})_{K\times K}$, with
$v_{ij}=1$ if $i=j$ and $v_{ij}=1/2$ if $i\neq j$.
\end{theorem}
\begin{remark}
    The Simes test is widely recognized as one of the most common multiplicity adjustments for handling positively related or mutually independent test statistics. In this context, our Theorems 1-2 provide test statistics $S_1, \ldots, S_K$ that are positively correlated, when combined with the Simes test, may effectively control the overall family-wise Type I error rate. 
\end{remark}
\subsection{Test Size under Some Commonly Used Models} \label{subsec: special cases}
In Section \ref{subsec: Theorems}, we established the asymptotic joint distribution of Wald test statistics $S_1, \ldots, S_K$ under a general framework of GLM (\ref{mod:working}) and CAR. Now we examine several important special cases of GLMs that are commonly encountered in practice. For each case, we derive the specific form of their asymptotic properties by applying the general theory to the corresponding canonical link functions. Specifically, we consider the logistic regression for binary outcomes, linear regression model for continuous outcome, Poisson regression for count data, and exponential model.\par
\noindent
\textbf{Logistic Regression:} 
In clinical trials, the logistic regression model is widely used because binary responses are common. For example, it is commonly used as a surrogate endpoint in oncology trials, such as overall response rate (ORR). Under the logistic regression, the conditional distribution of $Y$ given $\theta$ is a Bernoulli distribution with $\phi = 1$. When CAR procedures whose within-stratum imbalances are bounded in probability are implemented, we have 
\begin{align*}
    \left(S_1, S_2, \ldots, S_K \right) \stackrel{D}{\longrightarrow} \mathcal{N}\left(\bm{0}, \frac{\mathbb{E}[\text{Var}(Y \mid Z)]}{\text{Var}(Y)} \bm{V} \right). 
\end{align*}

Because we have $\operatorname{Var}(Y)=\mathbb{E}[\operatorname{Var}(Y \mid Z)]+\operatorname{Var}(\mathbb{E}[Y \mid Z])$ and $\operatorname{Var}(\mathbb{E}[Y \mid Z]) > 0$, it follows that $\operatorname{Var}(Y) > \mathbb{E}[\operatorname{Var}(Y \mid Z)]$. Hence, the test is conservative. \par
\noindent
\textbf{Linear Regression:} 
Linear regression is widely used for continuous endpoints in many therapeutic areas. Here, the dispersion parameter $\phi = \sigma^2$, which is usually unknown and could be estimated via the observations. Under the linear regression, $h(\cdot)$ is an identity function with $h'[h^{-1}\{\mathbb{E}(Y)\}] = 1$. Hence, under CAR procedures with bounded within-stratum imbalances and known $\sigma^2$, we have
\begin{align*}
    \left(S_1, S_2, \ldots, S_K \right) \stackrel{D}{\longrightarrow} \mathcal{N}\left(\bm{0},  \bm{V}\right)
\end{align*}

However, the $\sigma^2$ is usually unknown and must be estimated based on the data. In practice, the model with omitted covariate tends to overestimate the variance of the treatment effect estimate \citep{shao2010theory, ma2015testing}. An inflated variance results in test statistics that are smaller in absolute value, which leads to conservative Type I error. \par
\noindent
\textbf{Poisson Regression:} 
As clinical trials frequently involve count outcomes, such as the number of adverse events or disease recurrences over a specified time period, Poisson regression serves as a fundamental tool for analyzing such data while accounting for their discrete, non-negative nature. The distribution of $Y|\theta$ is a Poisson distribution with $\phi = 1$. Hence, under the null hypotheses, we have $\phi h'[h^{-1}\{\mathbb{E}(Y)\}] = \mathbb{E}(Y) = \mathbb{E}(\mathbb{E}[Y|Z]) = \mathbb{E}[\text{Var}(Y|Z)]$. It follows that 
\begin{align*}
    \left(S_1, S_2, \ldots, S_K \right) \stackrel{D}{\longrightarrow} \mathcal{N}\left(\bm{0},  \bm{V}\right)
\end{align*}

A unique property of Poisson regression is that the variance equals the mean by definition. Therefore, even when the analysis model omits certain covariates, the marginal behavior of $Y$ still maintains the relationship: $\mathbb{E}[\operatorname{Var}(Y \mid Z)]=\phi h^{\prime}[h^{-1}\{\mathbb{E}(Y)\}]$. This equality ensures that the asymptotic variance of the treatment effect estimate is correctly specified, resulting in test statistics that follow a normal distribution with unit variance asymptotically. Consequently, the Type I error rate remains well-controlled, neither inflated nor conservative.\par 
\noindent
\textbf{Exponential Model:} The exponential model is also widely used in the post-randomization analysis in many clinical trials. For example, when clinical trials measure intervals between recurring events, such as hospital readmission or disease relapses, exponential regression offers a straightforward approach for modeling these gaps while accommodating varying follow-up times and patient-specific risk factors. In the exponential model, the response $Y$ given $\theta$ is exponentially distributed with $\phi=1$. Under the null hypothesis and canonical inverse link $h(\eta) = -1/\eta$, we have $\phi h'[h^{-1}\{\mathbb{E}(Y)\}]=(\mathbb{E}(Y))^2=[\mathbb{E}\{\mathbb{E}(Y|Z)\}]^2 $. For $Y\mid Z$ following an exponential distribution with rate parameter $\lambda(Z)$, the mean is $1/\lambda(Z)$ and variance is $1/[\lambda(Z)]^2$, thus $E[\text{Var}(Y|Z)] = E\{[E(Y|Z)]^2\}$. It follows that $\phi h'[h^{-1}\{\mathbb{E}(Y)\}] \leq E[\text{Var}(Y|Z)]$. 

\begin{align*}
    \left(S_1, S_2, \ldots, S_K \right) \stackrel{D}{\longrightarrow} \mathcal{N}\left(\bm{0},  \frac{\mathbb{E}[\text{Var}(Y \mid Z)]}{\{\mathbb{E}[\mathbb{E}(Y|Z)]\}^2} \bm{V}\right)
\end{align*}
where $\mathbb{E}[\text{Var}(Y|Z)] / \{\mathbb{E}[\mathbb{E}(Y|Z)]\}^2 \geq 1$. \\
Consequently, when using stratified randomization without fully accounting for stratification factors in the analysis, tests for treatment effects tend to be inflated, leading to inflated Type I error rates.\par

Practically, our results lead to the following recommendations. For CAR procedures with bounded within-stratum imbalances ($\sigma_h^2 = 0$), the standard Wald test is conservative for logistic regression, valid for Poisson regression with the canonical log link, and inflated for the exponential model. Under general CAR procedures satisfying Conditions 1--2, such as Pocock and Simon's minimization, the direction of
the size distortion is instead determined by the sign of $\mathbb{E}[\operatorname{Var}(Y \mid Z)] + (K+1)\sigma_h^2/2 - W$, where $W$ is the working variance term defined in Table \ref{tab:summary cases}. In particular, the Poisson Wald test is mildly inflated when $\sigma_h^2 > 0$,
which is consistent with the simulation results under minimization in Section \ref{sec: simu}. We therefore recommend the adjusted test whenever
exact asymptotic size is desired, and in the multi-arm setting we combine the resulting p-values through the Simes--Hochberg procedures of Section
\ref{subsec:control FWER}.

\subsection{Additional Results with non-canonical link}
\label{sec: additionalResults}

While canonical link functions in GLMs provide computational convenience, they may not always be the most appropriate choice for real data analysis. In many practical situations, non-canonical links can offer better model fit, sometimes more interpretable results. For example, in biomedical studies, the probit link may be more suitable than the canonical logit link when the underlying process follows a normal tolerance distribution. Similarly, the complementary log-log link might be more appropriate when analyzing time-to-event data in discrete time. Additionally, some non-canonical links, such as the identity link for count data, can provide more direct interpretation of the effects. Therefore, extending our theoretical results to non-canonical links could broaden the practical applicability of our testing procedure while maintaining its statistical robustness. Hence, in this section, Theorems \ref{thm1}-\ref{thm2:sigmah} which were derived under canonical link functions are extended to the GLMs with non-canonical link functions. 

In GLMs with general link functions, the relationship between the natural parameter $\theta$ and the linear predictor $\eta = \delta_0 + \delta_1T_1 + \ldots + \delta_K T_K$ is characterized by $\theta = \gamma(\eta)$, where $\gamma$ is some function. Unlike canonical links where $\gamma$ is the identity function (i.e., $\theta = \eta$), non-canonical links involve a more general transformation that connects the linear predictor to the natural parameter of the exponential family.

\begin{theorem}
\label{thm3_noncanonical}
    Suppose that a covariate-adaptive randomization procedure satisfies the condition that all within-stratum imbalances are bounded in probability; then under the global null hypothesis $H_{0}: \delta_{1} = \cdots = \delta_{K} = 0$, as $N \to \infty$, we have
    \begin{align*}
    (S_1, \ldots, S_K)  \xrightarrow{D} \mathcal{N}\left(\bm{0}, \frac{\mathbb{E}[\text{Var}(Y \mid Z)]}{\phi h'[h^{-1}\{\mathbb{E}(Y)\}] / \gamma'[h^{-1}\{\mathbb{E}(Y)\}]}\bm{V}\right)
\end{align*}
where $\bm{V}$ is a K-dimensional square matrix with $v_{ij} = 1$ if $i = j$ and $v_{ij} = 1/2$ if $i \neq j$. 
\end{theorem}




\begin{table}[t!]
    \centering
    \caption{Summary of Type I error behavior of the standard Wald test across GLMs
    when covariates are omitted from the working model, under CAR procedures with
    bounded within-stratum imbalances (Theorems \ref{thm1} and \ref{thm3_noncanonical}).
    Here $\mu_Z = \mathbb{E}(Y \mid \bm{Z})$, and the working variance term is
    $W = \phi\, h'[h^{-1}\{\mathbb{E}(Y)\}] \big/ \gamma'[h^{-1}\{\mathbb{E}(Y)\}]$,
    which reduces to $\phi\, h'[h^{-1}\{\mathbb{E}(Y)\}]$ for canonical links.
    The test is conservative, valid, or inflated according to whether
    $\mathbb{E}[\operatorname{Var}(Y \mid Z)] - W$ is negative, zero, or positive.}
    \label{tab:summary cases}
    \resizebox{\textwidth}{!}{
    \begin{tabular}{l|l|c|c|l}
    \hline
        GLM Model & Link Function & Working Variance Term $W$
        & $\mathbb{E}[\operatorname{Var}(Y \mid Z)] - W$
        & Type I Error Behavior \\
        \hline
        Logistic Reg. & Logit (canonical)
        & $\mathbb{E}(Y)\{1-\mathbb{E}(Y)\}$
        & $-\operatorname{Var}(\mu_Z)$
        & Conservative \\
        Binary Reg. & Probit
        & $\mathbb{E}(Y)\{1-\mathbb{E}(Y)\}$\,$^{a}$
        & $-\operatorname{Var}(\mu_Z)$
        & Conservative \\
        Binary Reg. & Complementary Log-Log
        & $\mathbb{E}(Y)\{1-\mathbb{E}(Y)\}$\,$^{a}$
        & $-\operatorname{Var}(\mu_Z)$
        & Conservative \\
        Linear Reg. & Identity (canonical)
        & $\sigma^2$ (estimated)\,$^{b}$
        & $-\operatorname{Var}(\mu_Z)$
        & Conservative \\
        Poisson Reg. & Log (canonical)
        & $\mathbb{E}(Y)$
        & $0$
        & Valid \\
        Exponential & Inverse (canonical)
        & $\{\mathbb{E}(Y)\}^2$
        & $\operatorname{Var}(\mu_Z)$
        & Inflated \\
        Negative Binomial Reg. & Log
        & $\mathbb{E}(Y)+\alpha\{\mathbb{E}(Y)\}^2$\,$^{c}$
        & $\alpha\operatorname{Var}(\mu_Z)$
        & Inflated \\
        Gamma Reg. & Log
        & $\phi\{\mathbb{E}(Y)\}^2$
        & $\phi\operatorname{Var}(\mu_Z)$
        & Inflated \\
        Ordinal Endpoint\,$^{d}$ & Logit (per threshold $j$)
        & $\mathbb{E}(\pi_j)\{1-\mathbb{E}(\pi_j)\}$
        & $-\operatorname{Var}\{\mathbb{P}(Y \le j \mid \bm{Z})\}$
        & Conservative \\
        \hline
    \end{tabular}}
    \par
    \begin{minipage}{\textwidth}
    \vspace{2pt}
    \footnotesize
    $^{a}$\,For binary endpoints, $W$ equals the Bernoulli variance function evaluated
    at the marginal mean for any link, since
    $W = \phi\, b''\{\gamma(\eta_0)\}$ with $\eta_0 = h^{-1}\{\mathbb{E}(Y)\}$;
    hence the probit and complementary log-log links yield the same $W$ as the logit link.
    $^{b}$\,For linear regression, the dispersion estimated under the working model
    converges to $\operatorname{Var}(Y) = \mathbb{E}[\operatorname{Var}(Y \mid Z)]
    + \operatorname{Var}(\mu_Z)$, which overestimates
    $\mathbb{E}[\operatorname{Var}(Y \mid Z)]$.
    $^{c}$\,Negative binomial with $\phi = 1$ and fixed overdispersion parameter
    $\alpha > 0$, so that $\operatorname{Var}(Y \mid \bm{Z}) = \mu_Z + \alpha\mu_Z^2$.
    $^{d}$\,Refers to a separately analyzed binary indicator
    $\mathbb{I}(Y \le j)$ at a fixed threshold $j$, with cumulative probability
    $\pi_j = \mathbb{P}(Y \le j)$; the jointly fitted proportional-odds model is not
    covered by Theorems \ref{thm1}--\ref{thm4:sigmah_noncanonical}.
    All entries assume $\operatorname{Var}(\mu_Z) > 0$, i.e., the omitted covariates
    are prognostic. Under general CAR procedures satisfying Conditions 1--2, the
    comparison is between $\mathbb{E}[\operatorname{Var}(Y \mid Z)] + (K+1)\sigma_h^2/2$
    and $W$ (Theorems \ref{thm2:sigmah} and \ref{thm4:sigmah_noncanonical}).
    \end{minipage}
\end{table}


\begin{theorem}
\label{thm4:sigmah_noncanonical}
Suppose that a covariate-adaptive randomization procedure satisfies
Conditions 1 and 2. Then, under the global null hypothesis
$H_0:\delta_1=\cdots=\delta_K=0$, as $N\to\infty$,
\[
(S_1,\ldots,S_K)
\xrightarrow{D}
\mathcal{N}\left(
\bm{0},
\frac{
\mathbb{E}\!\left[\operatorname{Var}(Y\mid Z)\right]
+(K+1)\sigma_h^2/2
}{
\displaystyle
\phi
h'\!\left[
h^{-1}\{\mathbb{E}(Y)\}
\right]
/
\gamma'\!\left[
h^{-1}\{\mathbb{E}(Y)\}
\right]
}
\bm{V}
\right),
\]
where $\bm{V}=(v_{ij})_{K\times K}$, with
$v_{ij}=1$ if $i=j$ and $v_{ij}=1/2$ if $i\neq j$.
\end{theorem}

\begin{remark}
     Let $\eta=\delta_0+\delta_1 T_1+\ldots+\delta_K T_K$, where at most one of $T_1, \ldots, T_K$ equals 1 and the others are 0 (all equal 0 for the
     control arm). For binary endpoints, besides the canonical logit link $\log \{\mu /(1-\mu)\}=\delta_0+\delta_k T_k$, the probit link
     $\Phi^{-1}(\mu)=\delta_0+\delta_k T_k$ and the complementary log-log link $\log \{-\log (1-\mu)\}=\delta_0+\delta_k T_k$ are common
     alternatives. For Poisson regression modeling count data, although the canonical link is $\log (\mu)=\delta_0+\delta_k T_k$, the identity link
     $\mu=\delta_0+\delta_k T_k$ and the square-root link $\sqrt{\mu}=\delta_0+\delta_k T_k$ are sometimes preferred for direct interpretation. In gamma regression for positive continuous responses, while the inverse link $1 / \mu=\delta_0+\delta_k T_k$ is canonical, the log link $\log (\mu)=\delta_0+\delta_k T_k$ is widely used as it ensures positive fitted values. Here $\mu$ represents the conditional mean of the response given the treatment assignment, and these non-canonical options expand the flexibility of GLMs while maintaining the theoretical
     properties we established under mild regularity conditions.
\end{remark}

\section{Constructing Valid Test Statistics} 
\label{sec: new test}
Theorem \ref{thm1}, along with the discussion in Section \ref{subsec: Theorems} for common GLMs, shows that under CAR, the standard Wald test statistics fail to converge to a standard multivariate normal distribution in the absence of a fully adjusted model. This deviation arises from two primary sources: (1) The implementation of CAR and the potential working model misspecification can lead to test statistics that fail to converge to the correct asymptotic distribution under the null hypothesis $H_0$.
(2) The presence of multiple treatment arms and their associated subject assignment mechanisms induces dependence among test statistics, which subsequently affects multiple testing corrections. While existing methods can adjust for dependent test statistics or $p$-values, they typically yield conservative family-wise Type I error. To address these limitations, we propose the adjusted test statistics that converge to $N(\bm{0}, \bm{I})$ under the working model and the global null hypothesis $H_0: \delta_1 = \cdots = \delta_K = 0$, regardless of potential covariate omission. Such a test would simultaneously achieve robustness to model misspecification and valid family-wise Type I error control.

The scenarios where the within-stratum imbalances are bounded in probability, such as STR-PB, STR-BCD, and HuHuCAR are considered first. It should be noted that in HuHuCAR \citep{HuHuASYMPTOTIC}, the within-stratum imbalances remain bounded when non-zero weights are assigned to stratum-level balance. Based on the proof of Theorem \ref{thm1} in the Supplementary material, under the working model and the global null hypothesis $H_0: \delta_1 = \cdots = \delta_K = 0$, we can readily establish that:
\begin{align}
    \sqrt{N}\left(\frac{h'[h^{-1}\{\mathbb{E}(Y)\}]}{\sqrt{2(K+1)\mathbb{E}[\text{Var}(Y \mid Z)] }}\right)(\bm{V})^{-1/2} \hat{\bm{\delta}} \stackrel{D}{\longrightarrow} \mathcal{N}(\bm{0}, \bm{I}). 
\end{align}
Hence the newly constructed statistics are 
\begin{align}
\label{equ:adjS}
\bm{S}_{adj} &= \sqrt{N} \frac{h'\{h^{-1}(\bar{Y})\}}{\sqrt{2(K+1)\sigma^2_{\text{pooled}}}} (\bm{V})^{-1/2} \hat{\bm{\delta}}
\end{align}

where the pooled variance estimator is
\begin{align}
\sigma^2_{\text{pooled}} = \sum_{s=1}^S \frac{n_s}{N} \cdot
\frac{\sum_{i=1}^N (y_i - \widehat{\nu}_s)^2 \mathbf{1}\{\bm{z}_i = s\}}{n_s - 1},
\qquad
\widehat{\nu}_s = \frac{\sum_{i=1}^N y_i \mathbf{1}\{\bm{z}_i = s\}}{n_s}.
\end{align}
Here $\bar{Y}$ is the overall sample mean, $S$ is the total number of strata,
$n_s$ is the number of subjects in stratum $s$, and $\widehat{\nu}_s$ is the
sample mean of the responses in stratum $s$. The indicator
$\mathbf{1}\{\bm{z}_i = s\}$ takes value 1 if observation $i$ is in stratum
$s$, and 0 otherwise. The pooled variance estimator $\sigma_{\text{pooled}}^2$
weights each stratum-specific variance by the relative stratum size $n_s/N$.

\begin{remark}
    Although the stratification covariates (or stratum labels) are used in our variance estimator, we do not include them in the working model. Here strata are used only for variance calibration to reflect the covariance structure induced by covariate-adaptive randomization. This is different from modeling $E(Y\mid T,Z)$, which would require fitting a high-dimensional GLM with many stratum indicators and can be numerically unstable when strata are numerous or sparse.
\end{remark}

Similarly, for general CAR procedures under Conditions 1-2, the within-stratum-level imbalance may not be bounded. Under the working model and the null hypothesis $H_{0k}: \delta_k = 0 (k = 1, \ldots, K)$, we have:
\begin{align}
    \sqrt{N}\left(\frac{h'[h^{-1}\{\mathbb{E}(Y)\}]}{\sqrt{2(K+1)\left(\mathbb{E}[\text{Var}(Y \mid Z)] + (K+1) \sigma_h^2 / 2 \right) }}\right)(\bm{V})^{-1/2} \hat{\bm{\delta}} \stackrel{D}{\longrightarrow} \mathcal{N}(\bm{0}, \bm{I}). 
\end{align}
It follows that the adjusted test statistics are 
\begin{align}
\label{equ:adjS2}
\bm{S}_{adj} &= \sqrt{N} \frac{h'\{h^{-1}(\bar{Y})\}}{\sqrt{2(K+1)\sigma^2_{\text{pooled,adj}}}} (\bm{V})^{-1/2} \hat{\bm{\delta}}
\end{align}
where the adjusted pooled variance estimator is
\begin{align}
\sigma^2_{\text{pooled,adj}} = \sum_{s=1}^S \frac{n_s}{N} \frac{\sum_{i=1}^N (y_i - \widehat\nu_s)^2 \mathbf{1}\{\bm z_i = s\}}{n_s - 1} + \frac{(K+1) \hat{\sigma}_h^2}{2}
\end{align}
Compared with Equation \eqref{equ:adjS}, the only additional
quantity that needs to be estimated in Equation \eqref{equ:adjS2}
is $\sigma_h^2$. This additional variance arises because CAR
procedures such as Pocock and Simon's minimization do not generally
bound the within-stratum imbalances.

Under Condition 2, for each treatment-control contrast
$k=1,\ldots,K$, define
\begin{align}
H_{N,k}
&=
\frac{1}{\sqrt N}
\sum_{i=1}^N
(T_{ik}-T_{i0})\nu_i                                      \\
&=
\frac{1}{\sqrt N}
\sum_{s=1}^S
\{D_k(s)-D_0(s)\}\nu(s),
\end{align}
where
\[
\nu_i
=
\mathbb E(Y_i\mid\bm Z_i)
=
\nu(s)
\]
for a subject in stratum $s$. The parameter $\sigma_h^2$ is the
common variance parameter in the limiting distribution of
$H_{N,k}$. Therefore, it must be estimated by the variance of
$H_{N,k}$ across randomization replicates, rather than by the mean
of a signed imbalance statistic.

In practice, we estimate $\sigma_h^2$ using a fixed-covariate Monte
Carlo procedure as follows.

\textbf{Step 1:}
Keeping the observed covariates fixed, independently repeat the
specified CAR procedure $B$ times, producing treatment assignment
indicators
\[
T_{ik}^{(b)},
\qquad
b=1,\ldots,B.
\]

\textbf{Step 2:}
Estimate the conditional mean in stratum $s$ by
\[
\widehat\nu_s
=
\frac{
\sum_{i=1}^N
y_i\mathbf 1\{\bm z_i=s\}
}{
n_s
}.
\]
For each replicate $b$ and contrast $k=1,\ldots,K$, calculate
\begin{align}
\widehat H_{N,k}^{(b)}
&=
\frac{1}{\sqrt N}
\sum_{i=1}^N
\left(
T_{ik}^{(b)}-T_{i0}^{(b)}
\right)
\widehat\nu_i                                               \\
&=
\frac{1}{\sqrt N}
\sum_{s=1}^S
\left\{
D_k^{(b)}(s)-D_0^{(b)}(s)
\right\}
\widehat\nu_s,
\end{align}
where $\widehat\nu_i=\widehat\nu_s$ whenever $\bm z_i=s$.

\textbf{Step 3:}
For each treatment-control contrast, let
\[
\overline H_{N,k}
=
\frac{1}{B}
\sum_{b=1}^B
\widehat H_{N,k}^{(b)}
\]
and calculate the Monte Carlo variance
\[
\widehat\sigma_{h,k}^2
=
\frac{1}{B-1}
\sum_{b=1}^B
\left\{
\widehat H_{N,k}^{(b)}
-
\overline H_{N,k}
\right\}^2.
\]

\textbf{Step 4:}
Under treatment-label permutation invariance, the limiting
variances are identical across treatment-control contrasts.
Therefore, the pooled estimator is
\begin{equation}
\widehat\sigma_h^2
=
\frac{1}{K}
\sum_{k=1}^K
\widehat\sigma_{h,k}^2.
\end{equation}

\begin{remark}
\label{rmk:sigmah-consistency}
The estimator $\widehat\sigma_h^2$ is consistent under the following
conditions: $B = B_N \to \infty$; each stratum has positive probability,
$\mathbb{P}(\bm{Z} = s) > 0$, so that $\widehat\nu_s \xrightarrow{P} \nu(s)$
under $H_0$; and the conditional second moment of $H_{N,k}$ given the
covariates converges in probability to $\sigma_h^2$, which holds for Pocock
and Simon's minimization by the joint asymptotic normality of the
within-stratum imbalances \citep{hu2020theory, hu2023multi}. Under these
conditions, $\widehat\sigma_h^2 \xrightarrow{P} \sigma_h^2$, and the limiting
distribution of the adjusted statistics in \eqref{equ:adjS2} follows from
Slutsky's theorem.
\end{remark}

\begin{remark}
    The proposed adjusted test ensures valid Type I error control by accounting for the variance-covariance structure inherent in the dependency among treatment groups. While this adjustment ensures convergence to a standard multivariate normal distribution and addresses conservative Type I error in the standard Wald test, it does not always improve power. This is due to potential variance inflation, which can reduce sensitivity to detect the effects, especially for weak treatment effects or low correlation between test statistics. However, the adjusted test is advantageous in small samples, imbalanced group sizes, or highly correlated test statistics, where it provides robust error control and reliability in complex study designs. Thus, its application should balance Type I error control against potential power trade-offs, depending on the trial context.
\end{remark}

\section{Numerical Studies} 
\label{sec: simu}
In this section, we present simulation studies that evaluate our proposed adjusted test statistics in (\ref{equ:adjS}) and (\ref{equ:adjS2}). Our primary goal is to demonstrate that these tests maintain valid family-wise Type I error rates while potentially improving power under various CAR designs. We evaluate both the family-wise and individual Type I error rates for the original and adjusted tests when testing each null hypothesis $H_{0k}$ ($k = 1, \ldots, K$). The simulations cover three commonly used GLMs: logistic regression, Poisson regression, and exponential model with canonical link functions. For all simulation scenarios, we consider 3-arm trials with sample sizes of $N = 300$ and $N = 600$ for Type I error assessment, and $N = 300$ for power analysis. Each scenario is replicated 5000 times. \par

\begin{table}[t!]
    \caption{Simulated Type I error for Logistic Regression}
    \label{tab:logistypeIerror}
    \resizebox{\linewidth}{!}{
    \begin{tabular}{cc|ccc|ccc}
    \hline
     &   &  \multicolumn{3}{c}{\textbf{Wald Test}} &  \multicolumn{3}{c}{\textbf{Adjusted Test}}\\ 
              Sample Size  & Randomization & Family-wise & $H_1$ & $H_2$ &  \textcolor{black}{Family-wise} & $H_1$ & $H_2$ \\ 
                \hline 
                &  STR-PB  & 0.0190 & 0.0108 & 0.0114 & 0.0480 & 0.0270 & 0.0226 \\
        N = 300 &  PS   &  0.0140 & 0.0090 & 0.0101 &  0.0433 & 0.0259 & 0.0247 \\ 
                &  HuHuCAR & 0.0230 & 0.0116 & 0.0160 & 0.0534 & 0.0268 & 0.0298 \\ 
                 \hline
                &  STR-PB  &  0.0184 & 0.0102 & 0.0118 & 0.0480 & 0.0262 & 0.0234  \\
        N = 600 &  PS   &  0.0192 &  0.0109 & 0.0112 &  0.0476 &  0.0254 & 0.0272 \\ 
                &  HuHuCAR & 0.0258 & 0.0154 & 0.0156 & 0.0540  & 0.0298  & 0.0330 \\
                \hline
    \end{tabular}
    }
\end{table}

\begin{table}[t!]
    \caption{Simulated Type I error for Poisson Regression}
    \label{tab:poistypeIerror}
     \resizebox{\linewidth}{!}{
    \begin{tabular}{cc|ccc|ccc}
     \hline
    Sample Size & Randomization  &  \multicolumn{3}{c}{\textbf{Wald Test}} &  \multicolumn{3}{c}{\textbf{Adjusted Test}}\\ 
                &  & Family-wise & $H_1$ & $H_2$ &  \textcolor{black}{Family-wise} & $H_1$ & $H_2$ \\ 
                \hline 
                &  STR-PB  & 0.0525 & 0.0300 & 0.0310 & 0.0555 & 0.0275 & 0.0315 \\
        N = 300 &  PS   & 0.0570 & 0.0313 & 0.0304 & 0.0521 & 0.0282  & 0.0297 \\ 
                &  HuHuCAR &  0.0442 &  0.0294 &  0.0236 & 0.0530 & 0.0296 & 0.0264 \\ 
                   \hline
                &  STR-PB  &  0.0510 &  0.0295 & 0.0310 & 0.0540 & 0.0265 &  0.0320\\
        N = 600 &  PS   & 0.0574 & 0.0304  & 0.0322 & 0.0529 & 0.0272 & 0.0299 \\ 
                &  HuHuCAR & 0.0440 & 0.0238 & 0.0270 & 0.0474  & 0.0226  & 0.0262 \\
                 \hline
    \end{tabular}
    }
\end{table}

\begin{table}[t!]
    \caption{Simulated Type I error for Exponential Model}
    \label{tab:expotypeIerror}
    \resizebox{\linewidth}{!}{
    \begin{tabular}{cc|ccc|ccc}
     \hline
    Sample Size & Randomization  &  \multicolumn{3}{c}{\textbf{Wald Test}} &  \multicolumn{3}{c}{\textbf{Adjusted Test}}\\ 
                &  & Family-wise & $H_1$ & $H_2$ &  \textcolor{black}{Family-wise} & $H_1$ & $H_2$ \\ 
                \hline 
                &  STR-PB  & 0.1554 & 0.0956 & 0.0954 & 0.0538 & 0.0284 & 0.0298 \\
        N = 300 &  PS   & 0.1655 & 0.1087 & 0.1149 & 0.0522 & 0.0294  & 0.0303 \\ 
                &  HuHuCAR & 0.1794 & 0.1096 & 0.1170 & 0.0568  & 0.0338 & 0.0344 \\ 
                \hline
                &  STR-PB  &  0.1666 & 0.1068 &  0.0998  & 0.0550 & 0.0314 & 0.0268 \\
        N = 600 &  PS   & 0.1825 &  0.1077 & 0.1092 & 0.0544 & 0.0301 & 0.0289  \\ 
                &  HuHuCAR & 0.1730 & 0.1044 & 0.1078 & 0.0572  & 0.0296  & 0.0306 \\ 
                 \hline
    \end{tabular}
    }
\end{table}

In Tables \ref{tab:logistypeIerror}-\ref{tab:expotypeIerror}, we present the family-wise and individual Type I errors for commonly used GLMs under both the standard Wald test statistics and our proposed adjusted test statistics. The results demonstrate that the adjusted test statistics provide more efficient control of Type I errors.

For the treatment effect tests based on logistic regression, we set the parameters $\boldsymbol{\delta}=(-1,0,0)$ and $\beta=(1,2)$ with fixed sample sizes of $N=300$ and $N=600$. The covariates used in simulations include two binary variables $Z_1$ and $Z_2$, each evenly distributed with probability 0.5. Under this logistic regression setup, while the standard Wald test statistics yield conservative but acceptable family-wise Type I error rates, our adjusted test statistics demonstrate improved performance. Specifically, they maintain valid family-wise Type I error control and reduce conservativeness in individual Type I error rates for each test. Therefore, the adjusted tests show superior performance in controlling Type I errors.

For the treatment effect tests based on Poisson regression, we set $\bm{\delta} =(0.8,0,0)$ and $\beta=(0.25, 0.25)$ with sample size $N=300$ and $N = 600$.  The covariates $Z_1$, $Z_2$ are assumed to follow the same distribution as in the logistic regression scenario.  Under this setting, Table \ref{tab:poistypeIerror} shows that both the standard Wald test and our adjusted test maintain good control of Type I error rates. The family-wise error rate stays close to the nominal level of $0.05$ , and the individual Type I errors are also well controlled. This indicates that the asymptotic approximation works well for Poisson regression even without adjustment, likely due to the natural mean-variance relationship in the Poisson distribution. Nevertheless, our adjusted test statistics provide comparable performance while maintaining theoretical consistency with our general framework.

For the exponential model scenario, we specify parameters $\boldsymbol{\delta}=(0.2,0,0)$ and $\boldsymbol{\beta}=(0.5,1)$. The covariates $Z_1$ and $Z_2$ are also generated following the same distribution described earlier in the logistic regression setting. Empirical Type I errors from this setup are summarized in Table $\ref{tab:expotypeIerror}$. It can be seen that the standard Wald test leads to substantial inflation in both family-wise and individual Type I error rates, which are not acceptable in practical trial designs. Conversely, the adjusted test controls the Type I errors well. 

Having established the Type I error control properties of both standard Wald test and adjusted test, we now evaluate their power performance through simulation studies. We consider a 3-arm trial setting with one control group and two different dosages of a new drug. To create a realistic scenario for power analysis, we designate one treatment arm as having an optimal effect and the other as having a suboptimal effect compared to the control. This design allows us to assess how well our testing procedures can detect both strong and moderate treatment effects, which is particularly relevant in dose-finding studies where identifying the most effective dose level is crucial.

For power analysis, we consider trials with 300 subjects ($N=300$) using the HuHuCAR, which balances covariates among arms at stratum, marginal, and overall levels. While we present results for this specific CAR method, similar power performance was observed with other CAR procedures, including STR-PB, stratified biased coin design, and Pocock and Simon minimization. In the following tables, $\delta_1$ and $\delta_2$ denote the treatment effects of the two treatment arms relative to control.

\begin{table}[t!]
\caption{Simulated power for logistic regression with $\mathrm{N}=300$ (varying $\delta_1$, fixed $\delta_2$ )}
\label{tab:logis_power_1}
\resizebox{\linewidth}{!}{
\begin{tabular}{cccccc}
  \hline
$\delta_1$ & $\delta_2$ &  Test (Trt1) &   Test (Trt2) & Adjusted Test (Trt1) & Adjusted Test (Trt2) \\ 
  \hline
0.2 & 0.15 & 0.0296 & 0.0238 & 0.0484 & 0.0384 \\ 
0.4 & 0.15 &   0.0906 & 0.0288 & 0.1652 & 0.0332 \\ 
0.6 & 0.15 &   0.2270 & 0.0354 & 0.3640 & 0.0386 \\ 
0.8 & 0.15 &   0.4276 & 0.0392 & 0.6204 & 0.0476 \\ 
1.0 & 0.15 &   0.6404 & 0.0404 & 0.8148 & 0.0646 \\ 
1.2 & 0.15 &   0.8202 & 0.0422 & 0.9414 & 0.0918 \\ 
   \hline
\end{tabular}
}
\end{table}

\begin{table}[t!]
\caption{Simulated power for logistic regression with $\mathrm{N}=300$ (parallel increasing effects)}
\label{tab:logis_power_2}
\resizebox{\linewidth}{!}{
\begin{tabular}{cccccc}
  \hline
$\delta_1$ & $\delta_2$ &  Test (Trt1) &   Test (Trt2) & Adjusted Test (Trt1) & Adjusted Test (Trt2) \\ 
  \hline
0.2 & 0.15 & 0.0296 & 0.0238 & 0.0484 & 0.0384 \\ 
  0.4 & 0.35 & 0.0998 & 0.0848 & 0.1230 & 0.0888 \\ 
  0.6 & 0.55 & 0.2596 & 0.2252 & 0.2552 & 0.1926 \\ 
  0.8 & 0.75 & 0.4932 & 0.4502 & 0.4270 & 0.3574 \\ 
  1.0 & 0.95 & 0.7092 & 0.6736 & 0.6082 & 0.5328 \\ 
  1.2 & 1.15 & 0.8790 & 0.8566 & 0.7646 & 0.7118 \\ 
   \hline
\end{tabular}
}
\end{table}

Tables \ref{tab:logis_power_1} and \ref{tab:logis_power_2} present power analyses under the logistic regression model in two distinct scenarios. In Table \ref{tab:logis_power_1}, we fix $\delta_2$ at 0.15 while varying $\delta_1$ from 0.2 to 1.2. As $\delta_1$ increases, both tests show monotonically increasing power for Treatment 1, with our adjusted test consistently achieving higher power (e.g., 0.8148 versus 0.6404 when $\delta_1=1.0$). Meanwhile, the power for Treatment 2 remains stable, reflecting its fixed effect size.
Table \ref{tab:logis_power_2} examines a different scenario where both treatment effects increase in parallel, with $\delta_2$ maintained slightly below $\delta_1$. In this case, we observe a distinct pattern: while our adjusted test exhibits superior power for smaller treatment effects, the standard Wald test demonstrates better power when both treatment effects become large.

\begin{table}[t!]
\caption{Simulated power for Poisson regression with N = 300}
\label{tab:poisson_power}
\resizebox{\linewidth}{!}{
\begin{tabular}{cccccc}
  \hline
$\delta_1$ & $\delta_2$ &  Test (Trt1) &   Test (Trt2) & Adjusted Test (Trt1) & Adjusted Test (Trt2) \\ 
  \hline
  0.1 & 0.03& 0.0658 & 0.0338 & 0.0654 & 0.0306 \\ 
  0.2 & 0.06 & 0.1876 & 0.0518 & 0.1838 & 0.0362 \\ 
  0.3 & 0.09 & 0.3862 & 0.0802 & 0.3922 & 0.0416 \\ 
  0.4 & 0.12 & 0.6138 & 0.1176 & 0.6246 & 0.0520 \\ 
  0.5 & 0.15 & 0.8072 & 0.1644 & 0.8172 & 0.0608 \\ 
  0.6 & 0.18 & 0.9222 & 0.2164 & 0.9296 & 0.0664 \\ 
   \hline
\end{tabular}
}
\end{table}

In Table \ref{tab:poisson_power}, we evaluate the power performance for Poisson regression where $\delta_2$ is set to be proportional to $\delta_1\left(\delta_2=0.3 \delta_1\right)$. For Treatment 1 , both the standard Wald test and adjusted test provide comparable power as $\delta_1$ increases from 0.1 to 0.6, with the adjusted test showing slightly higher power (e.g., 0.6246 vs 0.6138 when $\delta_1=0.4$ ). For Treatment 2, where $\delta_2$ takes smaller values, the standard Wald test demonstrates marginally higher power than the adjusted test, though both maintain relatively low power due to the smaller values of $\delta_2$. 

\begin{table}[t!]
\centering
\caption{Simulated power for exponential model with $N = 300$}
\label{tab:expo_power}
    \begin{tabular}{cccc}
  \hline
 $\delta_1$ & $\delta_2$ & Adjusted Test (Trt1) & Adjusted Test (Trt2) \\ 
  \hline
0.05 & 0.03 & 0.0878 & 0.0472 \\ 
  0.10 & 0.06 & 0.2236 & 0.0692 \\ 
  0.15 & 0.09 & 0.3890 & 0.1000 \\ 
  0.20 & 0.12 & 0.5576 & 0.1416 \\ 
  0.25 & 0.15 & 0.7022 & 0.1854 \\
  0.30 & 0.18 & 0.8068 & 0.2270 \\
  0.35 & 0.21 & 0.8802 & 0.2752 \\ 
  0.40 & 0.24 & 0.9258 & 0.3210 \\ 
   \hline
\end{tabular}
\end{table}
For the exponential model, since we have demonstrated that the standard Wald test with Simes correction fails to control the family-wise Type I error rate, we focus only on examining the power of our adjusted test. Table \ref{tab:expo_power} presents the power analysis under our adjusted test where $\delta_2$ is proportional to $\delta_1\left(\delta_2=0.6 \delta_1\right)$. As $\delta_1$ increases from 0.05 to 0.40 , the power for Treatment $1$ increases substantially from 0.0878 to 0.9258 . The power for Treatment $2$ , though consistently lower due to smaller values of $\delta_2$, also shows a steady increase from 0.0472 to 0.3210 . These results show that our adjusted test maintains good power for detecting treatment differences while preserving valid Type I error control.

\section{Analysis of Metastatic Breast Cancer Trial} 
\label{sec: realtrial}
\cite{hamilton2022nextmonarch} reported a clinical trial (nextMONARCH) evaluating the efficacy of Abemaciclib, a cyclin-dependent kinase 4 and 6 (CDK4/6) inhibitor for treating metastatic breast cancer. The trial, registered under ClinicalTrials.gov Identifier \href{https://clinicaltrials.gov/study/NCT02747004}{NCT02747004}, was an open-label, randomized, controlled, phase 2 study conducted at 60 sites across 14 regions/countries from September 2016 to June 2017. A total of 234 eligible patients were enrolled and randomly assigned in a 1:1:1 ratio to one of three groups: abemaciclib (150 mg) plus tamoxifen (78 patients) (group A+T), abemaciclib (150 mg) alone (79 patients) (group A-150), or the reference treatment abemaciclib (200 mg) alone (77 patients) (group A-200). Randomization was stratified by the presence of liver metastases (Yes or No) and prior tamoxifen therapy (Yes or No) in the advanced/metastatic setting. The overall baseline involved patients with heavily pretreated hormone receptor-positive, HER2-negative metastatic breast cancer. The primary endpoint was investigator-assessed progression-free survival (PFS), while Objective Response Rate (ORR) served as a common binary endpoint to assess tumor response. Among all endpoints, ORR is defined as the proportion of patients with a complete or partial response. It is a key and common measure of treatment efficacy in cancer clinical trials, providing insight into the immediate tumor response to therapeutic intervention. In this trial, the ORR was 34.6\% for the abemaciclib plus tamoxifen group, 24.1\% for the abemaciclib 150 mg alone group, and 33.8\% for the abemaciclib 200 mg alone group. 

In this section, we conduct a post-randomization analysis using the newly proposed adjusted test. Our redesigned study utilizes the reported ORR of three treatments and the specified strata factors and their proportions to generate an increased number (N = 1000) of the individual patient data using a logistic regression model. Since the original paper did not provide the odds ratios for the strata factors, we assumed reasonable odds ratios based on published real-world and retrospective studies that provide information on ORR in HR+/HER2- metastatic breast cancer (mBC) patients, specifically examining the impact of liver metastases and prior tamoxifen use. 

For liver metastases, consistent evidence across major trials demonstrated its substantial impact on treatment response. In the MONARCH 3 trial \citep{goetz2017monarch}, patients with liver metastases showed a baseline ORR of 21.1\%. Similar findings were reported in MONALEESA-2 \citep{hortobagyi2016ribociclib} with 22\% ORR, and PALOMA-2 \citep{finn2016palbociclib} with 25.3\% ORR in liver metastasis subgroups. Based on these data, an odds ratio of 0.30 was selected for liver metastases, representing approximately a 70\% reduction in response odds.
For prior tamoxifen exposure, we broaden the scope to prior endocrine therapy exposure, as data specifically related to tamoxifen are limited. The impact of prior endocrine therapy exposure was demonstrated in several CDK4/6 inhibitor trials. In MONARCH 2 \citep{sledge2017monarch}, which evaluated abemaciclib plus fulvestrant in patients with HR+/HER2- advanced breast cancer, all patients had progressed on prior endocrine therapy, and prior tamoxifen use was a stratification factor due to its potential impact on treatment outcomes. Similarly, in PALOMA-3 \citep{turner2015palbociclib}, which included patients who had progressed on prior endocrine therapy, the response rates were analyzed according to the line and type of previous endocrine treatment, supporting its inclusion as a stratification factor. Based on these trials demonstrating the clinical relevance of prior endocrine therapy exposure, an odds ratio of 0.75 was selected as a conservative estimate for the effect of prior tamoxifen use in our model. Hence, the individual patient data are generated based on the following logistic regression: 
\begin{align*}
    \mathbb{P}(\text{Objective Response} = 1 \mid \bm{Z}) = \text{logit}^{-1}\Big(&-0.684 + 0.032 T_{\text{A+T}} -0.462 T_{\text{A-150}}\\& -1.20\ \mathbb{I}(\text{liver metastases} = \text{Yes})\\& -0.29\ \mathbb{I}(\text{prior tamoxifen usage} = \text{Yes})\Big)
\end{align*}
where $T_{\mathrm{A}+\mathrm{T}}$ is a binary indicator vector of length $N$ representing assignment to treatment group $\mathrm{A}+\mathrm{T}$. Specifically, each element of $T_{\mathrm{A}+\mathrm{T}}$ is coded as 1 if the patient is assigned to group $\mathrm{A}+\mathrm{T}$ and 0 otherwise. Similarly, $T_{\mathrm{A}-150}$ is a binary indicator vector indicating assignment to treatment group A-150, where each element is 1 for patients in group A-150 and 0 otherwise. In addition, we denote the covariate vector $Z$ here as including treatment indicators, liver metastasis status, and prior tamoxifen usage.\par
Based on the generated patient data, we apply a newly proposed adjusted inference method, which employs two-sided tests, to redesign and analyze the trial. Specifically, we are comparing the statistical power of the new inference method with that of the standard Wald test to demonstrate the improved performance of the new approach in detecting treatment effects within the context of the trial. 
\begin{table}[t!]
\centering
    \caption{Re-designed Metastatic Breast Cancer Trial (N = 1000)}
    \label{tab:Realtrial}
    \begin{tabular}{c|cccc}
    & \multicolumn{2}{c}{\textbf{Wald Test}} &  \multicolumn{2}{c}{\textbf{Adjusted Test}}\\ 
     Randomization & A+T  &  A-150  & A+T  &  A-150 \\ 
    \hline
        STR-PB & 0.0326 & 0.5438 & 0.1690 & 0.7652\\
        PS &  0.0332 & 0.5492 & 0.1712 & 0.7657\\
        HuHuCAR &  0.0346 & 0.5590 & 0.1720 & 0.7660\\ 
    \end{tabular}
\end{table}

\textcolor{black}{Table \ref{tab:Realtrial} shows that the adjusted test has higher power than the standard Wald test across all three randomization schemes (STR-PB, PS, and HuHuCAR). For A-150, which is the best experimental treatment arm, the improvement is larger. The adjusted test achieves about 0.76 power, compared to about 0.54 for the Wald test. These gains are consistent across randomization schemes, so they do not depend on the specific CAR procedure. Overall, the choice between the Wald and adjusted tests can materially affect the conclusion in this trial setting, since the Wald test is more likely to miss an effect that the adjusted test can detect.}


\section{Discussion}
\label{sec:discussion}
In this article, we proposed a family of robust tests for treatment effects in multi-arm trials within GLMs under CAR. Through the theoretical proof, extensive simulation studies and a real-world trial application, we demonstrate that the proposed inference method could maintain control of the Type I error rate, and also, in some cases, provide improved statistical power. \par 

From a theoretical perspective, we proved that the standard Wald test may not adequately control Type I error because the test statistics may not converge to the standard multivariate normal distribution under the working model and CAR. Our analysis revealed distinct performance patterns across commonly used models: conservative Type I error in logistic regression, valid control in Poisson regression, and inflation in exponential models. Notably, this represents the first theoretical investigation of asymptotic convergence for multi-arm trials under both CAR and GLMs, as previous research primarily focused on linear models. Our findings also demonstrated that the test statistics are correlated, leading us to propose adjusted test statistics that ensure convergence to the standard normal distribution.\par

Our simulation studies corroborated these theoretical findings. The proposed adjusted test maintained family-wise Type I error control across all scenarios, a crucial requirement for post-randomization analysis. This control was particularly important for GLMs that showed inflated family-wise Type I error under the standard Wald test. Power analysis revealed that our adjusted test provided enhanced performance in certain scenarios under the logistic regression model. For the exponential model, we focused solely on the adjusted test's power, as the Wald test's failure to maintain Type I error control made comparisons meaningless. The practical application to a metastatic breast cancer trial validated our theoretical findings and demonstrated potential power advantages, suggesting tangible benefits for real-world clinical trials.\par

\textcolor{black}{While our method performs well in the settings studied, several limitations are worth noting. First, our theoretical results are developed under CAR schemes that stratify on discrete covariates, so the current framework does not directly cover continuous stratification variables without additional modeling or discretization. Second, our analysis is conducted within a specified GLM family, and it relies on using the intended link function in the working model. Third, in scenarios with moderate or strong treatment effects, the adjustment may reduce power compared to the traditional Wald test due to variance inflation. However, when the standard Wald test is conservative or inflated, the adjusted test can provide more reliable inference by bringing Type I error closer to the nominal level. When the Wald test is conservative, this adjustment can also improve power at the nominal level, especially when treatment effects are small. Despite these nuances, the adjusted test remains a useful tool for multi-arm trials under CAR, especially when test statistics are correlated.}\par 

Several directions for future research emerge from this work. An immediate extension would be to expand the current framework beyond CAR designs to incorporate response-adaptive randomization. The existing framework focuses primarily on achieving covariate balance during randomization; however, with the growing emphasis on precision medicine, solely balancing covariates may no longer suffice. Future randomization schemes should also account for individual patient responses to improve trial efficiency and enhance outcome precision. The proposed framework is well-suited for such advancements and can be readily extended to covariate-adjusted response-adaptive randomization, enabling more dynamic and personalized trial designs. Additionally, while the adjusted test was developed within the GLM framework, similar principles could be applied to time-to-event endpoints using a Cox proportional hazards working model, expanding its applicability to survival analysis and other time-to-event scenarios. These extensions represent promising areas for future research and practical development.
\section*{Supplementary Materials}
Additional technical details are provided in the Supplementary Materials. They include the formal definitions of the imbalance measures under multi-arm CAR, auxiliary lemmas, proofs of the main theorems in Section~\ref{sec: theorem}, and derivations of the working variance term under different GLMs.\par



\bibhang=1.7pc
\bibsep=2pt
\fontsize{9}{14pt plus.8pt minus .6pt}\selectfont
\renewcommand\bibname{\large \bf References}
\expandafter\ifx\csname
natexlab\endcsname\relax\def\natexlab#1{#1}\fi
\expandafter\ifx\csname url\endcsname\relax
  \def\url#1{\texttt{#1}}\fi
\expandafter\ifx\csname urlprefix\endcsname\relax\def\urlprefix{URL}\fi

\bibliographystyle{chicago}      
\bibliography{sample2}   







\vskip .65cm
\noindent
Guannan Zhai, Feifang Hu
\vskip 2pt
\noindent
The George Washington University
\vskip 2pt
\noindent
E-mail: (guannanzhai@gwu.edu), (feifang@gwu.edu)
\vskip 2pt

\end{document}